\documentclass[a4paper,11pt]{article}
\usepackage{amsfonts}
\usepackage{amsmath}
\usepackage{amsthm}
\usepackage{amssymb}
\usepackage[ruled,vlined]{algorithm2e}
\usepackage[hyphens]{url}
\usepackage{hyperref}
\usepackage{graphicx}
\usepackage{subfigure}
\usepackage{xspace}
\usepackage{units}

\newcommand{\qq}[1]{``#1''}
\newcommand{\modelstyle}[1]{\textsc{#1}\xspace}
\newcommand{\fsynch}{\modelstyle{FSynch}}
\newcommand{\ssynch}{\modelstyle{SSynch}}
\newcommand{\asynch}{\modelstyle{ASynch}}
\newcommand{\wait}{\modelstyle{Wait}}
\newcommand{\look}{\modelstyle{Look}}
\newcommand{\compute}{\modelstyle{Compute}}
\newcommand{\move}{\modelstyle{Move}}
\newcommand{\rend}{\modelstyle{Rendezvous}}
\newcommand{\gath}{\modelstyle{Gathering}}
\newcommand{\conv}{\modelstyle{Convergence}}

\newcommand{\oneh}{\nicefrac{1}{2}\xspace}

\newcommand{\rra}{Theorem~\ref{r1}\xspace}
\newcommand{\rrb}{Theorem~\ref{r2}\xspace}
\newcommand{\rrc}{Lemma~\ref{r3}\xspace}
\newcommand{\rrd}{Lemma~\ref{r4}\xspace}
\newcommand{\rre}{Lemma~\ref{r5}\xspace}
\newcommand{\rrf}{Lemma~\ref{r6}\xspace}
\newcommand{\rrg}{Observation~\ref{r7}\xspace}
\newcommand{\rrh}{Lemma~\ref{r8}\xspace}

\newcommand{\rrj}{Corollary~\ref{r10}\xspace}
\newcommand{\rrk}{Theorem~\ref{r11}\xspace}
\newcommand{\rrl}{Proposition~\ref{r12}\xspace}
\newcommand{\rrm}{Lemma~\ref{r13}\xspace}
\newcommand{\rrn}{Proposition~\ref{r14}\xspace}
\newcommand{\rro}{Theorem~\ref{r15}\xspace}
\newcommand{\rrp}{Lemma~\ref{r16}\xspace}
\newcommand{\rrq}{Lemma~\ref{r17}\xspace}
\newcommand{\rrr}{Lemma~\ref{r18}\xspace}
\newcommand{\rrs}{Lemma~\ref{r19}\xspace}
\newcommand{\rrt}{Theorem~\ref{r20}\xspace}

\newcommand{\rrw}{Lemma~\ref{r23}\xspace}
\newcommand{\rrx}{Theorem~\ref{r24}\xspace}
\newcommand{\rry}{Lemma~\ref{r25}\xspace}
\newcommand{\rrz}{Observation~\ref{r26}\xspace}

\newtheorem{theorem}{Theorem}[section]
\newtheorem{corollary}[theorem]{Corollary}
\newtheorem{proposition}[theorem]{Proposition}
\newtheorem{lemma}[theorem]{Lemma}
\newtheorem{observation}[theorem]{Observation}

\begin{document}
\title{Rendezvous of two robots with visible bits}
\author{Giovanni Viglietta\\
\texttt{viglietta@gmail.com}}

\maketitle

\begin{abstract}
We study the rendezvous problem for two robots moving in the plane (or on a line). Robots are autonomous, anonymous, oblivious, and carry colored lights that are visible to both. We consider deterministic distributed algorithms in which robots do not use distance information, but try to reduce (or increase) their distance by a constant factor, depending on their lights' colors.

We give a complete characterization of the number of colors that are necessary to solve the rendezvous problem in every possible model, ranging from fully synchronous to semi-synchronous to asynchronous, rigid and non-rigid, with preset or arbitrary initial configuration.

In particular, we show that three colors are sufficient in the non-rigid asynchronous model with arbitrary initial configuration. In contrast, two colors are insufficient in the rigid asynchronous model with arbitrary initial configuration and in the non-rigid asynchronous model with preset initial configuration.

Additionally, if the robots are able to distinguish between zero and non-zero distances, we show how they can solve rendezvous and detect termination using only three colors, even in the non-rigid asynchronous model with arbitrary initial configuration.
\end{abstract}

\section{Introduction}\label{s1}

\subsection{Models for mobile robots}
The basic robot model we employ has been thoroughly described in~\cite{gather,visbits,book,survey}. Robots are modeled as points freely moving in $\mathbb R^2$ (or $\mathbb R$). Each robot has its own coordinate system and its own unit distance, which may differ from the others. Robots operate in cycles that consist of four phases: \wait, \look, \compute, and \move.

In a \wait phase a robot is idle; in a \look phase it gets a snapshot of its surroundings (including the positions of the other robots); in a \compute phase it computes a destination point; in a \move phase it moves toward the destination point it just computed, along a straight line. Then the cycle repeats over and over.

Robots are anonymous and oblivious, meaning that they do not have distinct identities, they all execute the same algorithm in each \compute phase, and the only input to such algorithm is the snapshot coming from the previous \look phase.

In a \move phase, a robot may actually reach its destination, or it may be stopped before reaching it. If a robot always reaches its destination by the end of each \move phase, then the model is said to be \emph{rigid}. If a robot can unpredictably be stopped before, the model is \emph{non-rigid}. However, even in non-rigid models, during a \move phase, a robot must always be found on the line segment between its starting point and the destination point. Moreover, there is a constant distance $\delta>0$ that a robot is guaranteed to walk at each cycle. That is, if the destination point that a robot computes is at most $\delta$ away (referred to some global coordinate system), then the robot is guaranteed to reach it by the end of the next \move phase. On the other hand, if the destination point is more than $\delta$ away, the robot is guaranteed to approach it by at least $\delta$.

In the basic model, robots cannot communicate in conventional ways or store explicit information, but a later addition to this model allows each robot to carry a ``colored light'' that is visible to every robot (refer to~\cite{visbits}). There is a fixed amount of possible light colors, and a robot can compute its destination and turn its own light to a different color during a \compute phase, based on the light colors that it sees on other robots and on itself. Usually, when robots start their execution, they have all their lights set to a predetermined color. However, we are also interested in algorithms that work regardless of the initial color configurations of the robots.

In the fully synchronous model (\fsynch) all robots share a common notion of time, and all their phases are executed synchronously. The semi-synchronous model (\ssynch) is similar, but not every robot may be active at every cycle. That is, some robots are allowed to ``skip'' a cycle at unpredictable times, by extending their \wait phase to the whole cycle. However, the robots that are active at a certain cycle still execute it synchronously. Also, no robot can remain inactive for infinitely many consecutive cycles. Finally, in the asynchronous model (\asynch) there is no common notion of time, and each robot's execution phase may last any amount of time, from a minimum $\epsilon>0$ to an unboundedly long, but finite, time.

Figure~\ref{f0} shows all the possible models arising from combining synchronousness, rigidity, and arbitrarity of the initial light colors. The trivial inclusions between models are also shown.

\begin{figure}[ht]
\centering{\includegraphics[width=\linewidth]{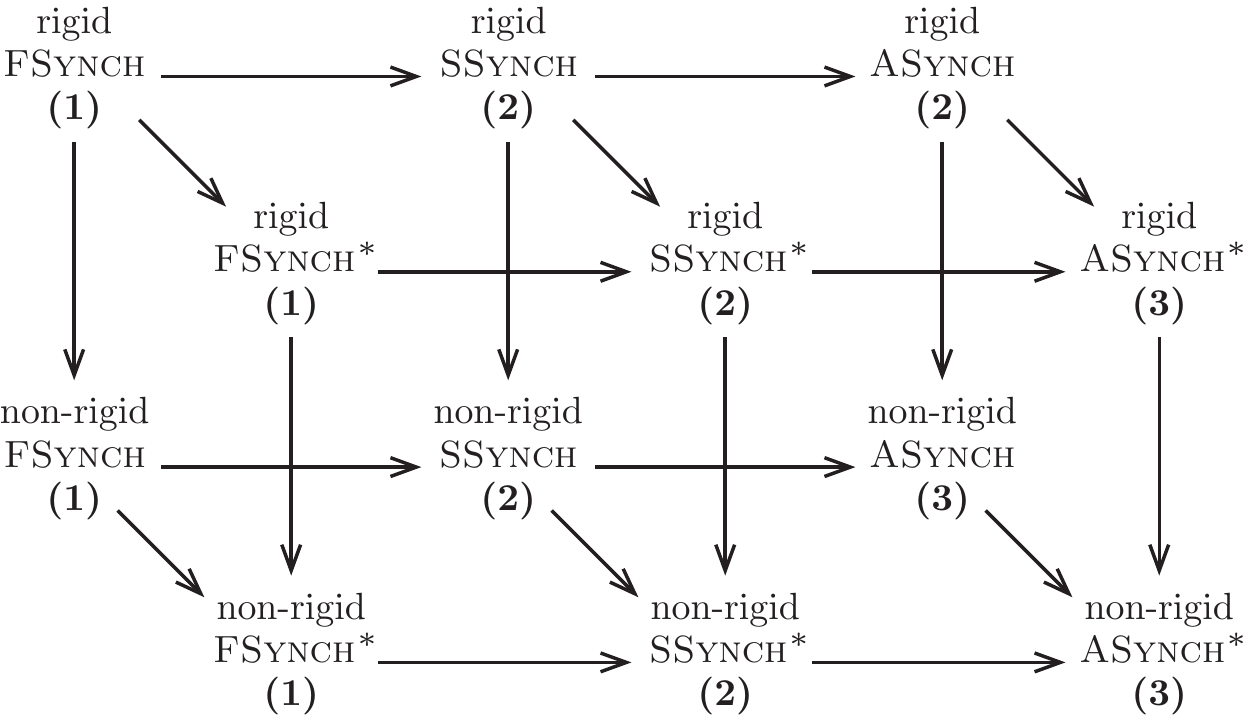}}
\caption{Robot models with their trivial inclusions. An asterisk means that the initial color configuration may be arbitrary; no asterisk means that it is fixed. The numbers indicate the minimum amounts of distinct colors that are necessary to solve \rend in each model (cf.~\rrt).}
\label{f0}
\end{figure}

Without loss of generality, in this paper we will assume \look phases in \asynch to be instantaneous, and we will assume that a robot's light's color may change only at the very end of a \compute phase.

\subsection{Gathering mobile robots}

\gath is the problem of making a finite set of robots in the plane reach the same location in a finite amount of time, and stay there forever, regardless of their initial positions. Such location should not be given as input to the robots, but they must implicitly determine it, agree on it, and reach it, in a distribute manner. Note that this problem is different from \conv, in which robots only have to approach a common location, but may never actually reach it.

For any set of more than two robots, \gath has been solved in non-rigid \asynch, without using colored lights (see~\cite{gather}). The special case with only two robots is also called \rend, and it is easily seen to be solvable in non-rigid \fsynch but unsolvable in rigid \ssynch, if colored lights are not used (see~\cite{ssynch}).

\begin{proposition}\label{r14}
If only one color is available, \rend is solvable in non-rigid \fsynch and unsolvable in rigid \ssynch.
\end{proposition}
\begin{proof}
In non-rigid \fsynch, consider the algorithm that makes each robot move to the midpoint of the current robots' positions. At each move, the distance between the two robots is reduced by at least $2\delta$, until it becomes less than $2\delta$, and the robots gather.

Suppose that an algorithm exists that solves \rend in rigid \ssynch by using just one color. Let us assume that the two robots' axes are oriented symmetrically, in opposite directions. This implies that, if we activate both robots at each cycle, they obtain isometric snapshots, and thus they make moves that are symmetric with respect to their current midpoint. Therefore, by doing so, the robots can never meet unless they compute the midpoint. If they do it, we just activate one robot for that cycle (and each time this happens, we pick a different robot, alternating). As a result, the robots never meet, regardless of the algorithm.
\end{proof}

However, in~\cite{visbits} it was shown how \rend can be solved even in non-rigid \asynch using lights of four different colors, and starting from a preset configuration of colors. Optimizing the amount of colors was left as an open problem.

\subsection{Our contribution}

In this paper, we will determine the minimum number of colors required to solve \rend in all models shown in Figure~\ref{f0}, with some restrictions on the class of available algorithms.

Recall that robots do not necessarily share a global coordinate system, but each robot has its own. If the coordinate system of a robot is not even self-consistent (i.e., it can unpredictably change from one cycle to another), then the only reliable reference for each robot is the position of the other robot(s) around it. In this case, the only type of move that is consistent will all possible coordinate systems is moving to a linear combination of the robots' positions, whose coefficients may depend on the colored lights. In particular, when the robots are only two, we assume that each robot may only compute a destination point of the form
$$(1-\lambda)\cdot me.position + \lambda\cdot other.position,$$
for some $\lambda\in \mathbb R$. In turns, $\lambda$ is a function of $me.light$ and $other.light$ only. This class of algorithms will be denoted by $\mathcal L$.

In Section~\ref{s2}, we will prove that two colors are sufficient to solve \rend in non-rigid \ssynch with arbitrary initial configuration and in rigid \asynch with preset initial configuration, whereas three colors are sufficient in non-rigid \asynch with arbitrary initial configuration. All the algorithms presented are of class $\mathcal L$.

On the other hand, in Section~\ref{s3} we show that even termination detection can be achieved in non-rigid \asynch with arbitrary initial configuration using only three colors, although our algorithm is not of class $\mathcal L$ (indeed, no algorithm of class $\mathcal L$ can detect termination in \rend).

In contrast, in Section~\ref{s4} we prove that no algorithm of class $\mathcal L$ using only two colors can solve \rend in rigid \asynch with arbitrary initial configuration or in non-rigid \asynch with preset initial configuration.

Finally, in Section~\ref{s5} we put all these results together and we conclude with a complete characterization of the minimum amount of colors that are needed to solve \rend in every model (see \rrt).

\section{Algorithms for rendezvous}\label{s2}

\subsection{Two colors for the non-rigid semi-synchronous model}

For non-rigid \ssynch, we propose Algorithm~\ref{alg1}, also represented in Figure~\ref{f1}. Labels on arrows indicate the color that is seen on the other robot, and the destination of the next \move with respect to the position of the other robot. \qq{0} stands for \qq{do not move}, \qq{\oneh} means \qq{move to the midpoint}, and \qq{1} means \qq{move to the other robot}. The colors used are only two, namely $A$ and $B$.

\begin{algorithm}\label{alg1}
\caption{Rendezvous for non-rigid \ssynch and rigid \asynch}
\DontPrintSemicolon
$me.destination \longleftarrow me.position$\;
\If{$me.light = A$}{
	\If{$other.light = A$}{
		$me.light \longleftarrow B$\;
		$me.destination \longleftarrow (me.position+other.position)/2$\;
	}
	\Else{
		$me.destination \longleftarrow other.position$\;
	}
}
\ElseIf{$other.light = B$}{
	$me.light \longleftarrow A$\;
}
\end{algorithm}

\begin{figure}[ht]
\centering{\includegraphics[scale=1.5]{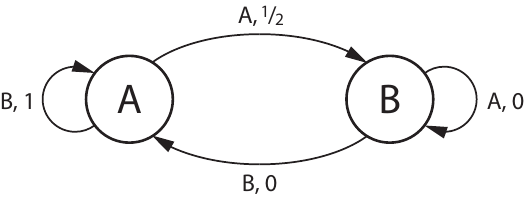}}
\caption{Illustration of Algorithm~\ref{alg1}}
\label{f1}
\end{figure}

\begin{lemma}\label{r3}
If the two robots start a cycle with their lights set to opposite colors, they eventually gather.
\end{lemma}
\begin{proof}
Both robots retain their colors at every cycle, and the $A$-robot keeps computing the other robot's location, while the $B$-robot keeps waiting. Hence, their distance decreases by at least $\delta$ for every cycle in which the $A$-robot is active, until the distance becomes smaller than $\delta$, and the robots gather.
\end{proof}

\begin{theorem}\label{r1}
Algorithm~\ref{alg1} solves \rend in non-rigid \ssynch, regardless of the colors in the initial configuration.
\end{theorem}
\begin{proof}
If the robots start in opposite colors, they gather by \rrc. If they start in the same color, they keep alternating colors until one robot is active and one is not. If this happens, they gather by \rrc. Otherwise, the two robots are either both active or both inactive at each cycle, and they keep computing the midpoint every other active cycle. Their distance decreases by at least $2\delta$ each time they move, until it becomes smaller than $2\delta$, and they finally gather.
\end{proof}

\subsection{Two colors for the rigid asynchronous model}

We prove that Algorithm~\ref{alg1} solves \rend in rigid \asynch as well, provided that the initial color is $A$ for both robots.

\begin{lemma}\label{r4}
If, at some time $t$, the two robots have opposite colors and neither of them is in a \compute phase that will change its color, they will eventually solve \rend.
\end{lemma}
\begin{proof}
Each robot retains its color at every cycle after time $t$, because it keeps seeing the other robot in the opposite color at every \look phase. As soon as the $A$-robot performs its first \look after time $t$, it starts chasing the other robot. On the other hand, as soon as the $B$-robot performs its first \look after time $t$, it stops forever. Eventually, the two robots will gather and never move again.
\end{proof}

\begin{theorem}\label{r2}
Algorithm~\ref{alg1} solves \rend in rigid \asynch, provided that both robots start with their lights set to $A$.
\end{theorem}
\begin{proof}
Let $r$ be the first robot to perform a \look. Then $r$ sees the other robot $s$ set to $A$, and hence it turns $B$ and computes the midpoint $m$. Then, as long as $s$ does not perform its first \look, $r$ stays $B$ because it keeps seeing $s$ set to $A$. Hence, if $s$ performs its first \look after $r$ has turned $B$, \rrd applies, and the robots will solve \rend.

On the other hand, if $s$ performs its first \look when $r$ is still set to $A$ (hence still in its starting location), $s$ will turn $B$ and compute the midpoint $m$, as well. If some robot reaches $m$ and performs a \look while the other robot is still set to $A$, the first robot waits until the other turns $B$. Without loss of generality, let $r$ be the first robot to perform a \look while the other robot is set to $B$. This must happen when $r$ is in $m$ and set to $B$, hence it will turn $A$ and stay in $m$. If $r$ turns $A$ before $s$ has reached $m$, then \rrd applies. Otherwise, $r$ turns $A$ when $s$ is already in $m$, and both robots will stay in $m$ forever, as they will see the other robot in $m$ at every \look.
\end{proof}

\subsection{Three colors for the non-rigid asynchronous model}

For non-rigid \asynch, we propose Algorithm~\ref{alg3}, also represented in Figure~\ref{f3}. The colors used are three, namely $A$, $B$, and $C$.

\begin{algorithm}\label{alg3}
\caption{Rendezvous for non-rigid \asynch}
\DontPrintSemicolon
$me.destination \longleftarrow me.position$\;
\If{$me.light = A$}{
	\If{$other.light = A$}{
		$me.light \longleftarrow B$\;
		$me.destination \longleftarrow (me.position+other.position)/2$\;
	}
	\ElseIf{$other.light = B$}{
		$me.destination \longleftarrow other.position$\;
	}
}
\ElseIf{$me.light = B$}{
	\If{$other.light = B$}{
		$me.light \longleftarrow C$\;
	}
	\ElseIf{$other.light = C$}{
		$me.destination \longleftarrow other.position$\;
	}
}
\Else{
	\If{$other.light = C$}{
		$me.light \longleftarrow A$\;
	}
	\ElseIf{$other.light = A$}{
		$me.destination \longleftarrow other.position$\;
	}
}
\end{algorithm}

\begin{figure}[ht]
\centering{\includegraphics[scale=1.5]{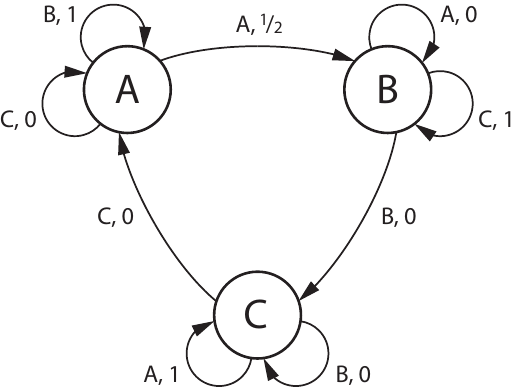}}
\caption{Illustration of Algorithm~\ref{alg3}}
\label{f3}
\end{figure}

\begin{observation}\label{r26}
A robot retains its color if and only if it sees the other robot set to a different color.
\end{observation}

\begin{lemma}\label{r25}
If, at some time $t$, the two robots are set to different colors, and neither of them is in a \compute phase that will change its color, they will eventually solve \rend.
\end{lemma}
\begin{proof}
The two robots keep seeing each other set to different colors, and hence they never change color, by \rrz. One of the two robots will eventually stay still, and the other robot will then approach it by at least $\delta$ at every \move phase, until their distance is less than $\delta$, and they gather. As soon as they have gathered, they will stay in place forever.
\end{proof}

\begin{theorem}\label{r24}
Algorithm~\ref{alg3} solves \rend in non-rigid \asynch, regardless of the colors in the initial configuration.
\end{theorem}
\begin{proof}
If the robots start the execution at different colors, they solve \rend by \rry.

If they both start in $A$, then let $r$ be the first robot to perform a \look. $r$ plans to turn $B$ and move to the midpoint. If it turns $B$ before the other robot $s$ has performed a \look, then \rry applies.

Otherwise, $s$ plans to turn $B$ and move to the midpoint, as well. If a robot stops and sees the other robot still set to $A$, it waits. 
Without loss of generality, let $r$ be the first robot to perform a \look and see the other robot set to $B$. $r$ now plans to turn $C$, but if it does so before $s$ has performed a \look, \rry applies.

So, let us assume that both robots have seen each other in $B$ and they both plan to turn $C$. Once again, if a robot turns $C$ and sees the other robot still in $B$, it waits. Without loss of generality, let $r$ be the first robot to see the other robot in $C$. $r$ plans to turn $A$, but if it does so before $s$ has performed a \look, \rry applies.

Assume that both robots see each other in $C$ and they both plan to turn $A$. If a robot turns $A$ and sees the other robot still in $C$, it waits. At some point, both robots are in $A$ again, in a \wait phase, but they have approached each other. They both moved toward the midpoint in their first cycle, and then they just made null moves. As a consequence, if their distance was smaller than $2\delta$, they have gathered. Otherwise, the distance has decreased by at least $2\delta$. As the execution goes on and the same pattern of transitions repeats, the distance keeps decreasing until the robots gather. As soon as they have gathered, they never move again, hence \rend is solved.

The cases in which the robots start both in $B$ or both in $C$ are resolved with the same reasoning. Note that all the states with both robots set to the same color and in a \wait phase have been reached in the analysis above.
\end{proof}

\section{Termination detection}\label{s3}

Suppose we wanted our robots to acknowledge that they have gathered, in order to turn off, or ``switch gears'' and start performing a new task.

\begin{observation}
If the model is \ssynch, termination detection is trivially obtained by checking at each cycle if the robots' locations coincide.
\end{observation}

Unfortunately, in \asynch, correct termination detection is harder to obtain. Observe that both Algorithm~\ref{alg1} (for rigid \asynch) and Algorithm~\ref{alg3} (for non-rigid \asynch) fail to guarantee termination detection. Indeed, suppose that robot $r$ is set to $A$ and sees the other robot $s$ set to $B$, and that the two robots coincide. Then $r$ cannot tell if $s$ is still moving or not. If $s$ is not moving, it is safe for $r$ to terminate, but if $s$ is moving, then $r$ has still to ``chase'' $s$, and cannot terminate yet.

To guarantee correct termination detection in non-rigid \asynch, we propose Algorithm~\ref{alg2}, also represented in Figure~\ref{f2}. Note that different rules may apply depending on the distance between the two robots, indicated by $d$ in the picture. However, robots need only distinguish between zero and non-zero distances. The colors used are again three, namely $A$, $B$, and $C$.

\begin{algorithm}\label{alg2}
\caption{Rendezvous for non-rigid \asynch with termination}
\DontPrintSemicolon
$me.destination \longleftarrow me.position$\;
\If{$me.light = A$}{
	\If{$other.light = A$}{
		\If{$other.position \neq me.position$}{
			$me.light \longleftarrow B$\;
			$me.destination \longleftarrow (other.position+me.position)/2$\;
		}
		\Else{
			$me.light \longleftarrow C$\;
		}
	}
	\ElseIf{$other.light = B$}{
		$me.destination \longleftarrow other.position$\;
	}
	\ElseIf{$other.position \neq me.position$}{
		$me.destination \longleftarrow other.position$\;
	}
	\Else{
		$me.light \longleftarrow C$\;
	}
}
\ElseIf{$me.light = B$}{
	\If{$other.light = A$ \textbf{\emph{and}} $other.position = me.position$}{
		$me.light \longleftarrow C$\;
	}
	\ElseIf{$other.light = B$}{
		$me.light \longleftarrow A$\;
	}
	\ElseIf{$other.position \neq me.position$}{
		$me.destination \longleftarrow other.position$\;
	}
	\Else{
		$me.light \longleftarrow C$\;
	}
}
\ElseIf{$other.light = C$}{
	\If{$other.position \neq me.position$}{
		$me.light \longleftarrow A$\;
	}
	\Else{
		terminate\;
	}
}
\end{algorithm}

\begin{figure}[ht]
\centering{\includegraphics[scale=1.5]{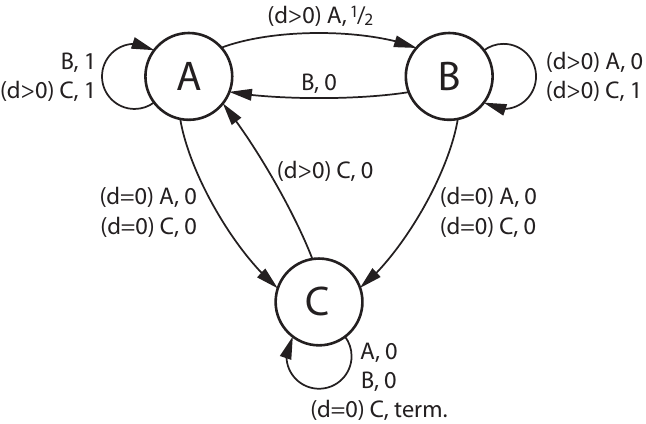}}
\caption{Illustration of Algorithm~\ref{alg2}}
\label{f2}
\end{figure}

\begin{observation}\label{r7}
No robot can move while it is set to $C$.
\end{observation}

\begin{lemma}\label{r6}
If some robot ever turns $C$ from a different color, the two robots will gather and their execution will terminate correctly.
\end{lemma}
\begin{proof}
A robot can turn $C$ only if it performs a \look while the other robot is in the same location. If robot $r$ performs a \look at time $t$ that makes it turn $C$, then $r$ stays $C$ forever after, unless it sees the other robot $s$ set to $C$ as well, in a different location. Let $t'>t$ be the first time this happens. Due to \rrg, $r$ does not move between $t$ and $t'$. On the other hand, $s$ coincides with $r$ at time $t$. Then $s$ must turn some other color and move away from $r$, and then turn $C$ at some time $t''$ such that $t<t''\leqslant t'$. But, in order to turn $C$, $s$ would have to coincide with $r$, which is a contradiction.

Hence $r$ will stay $C$ and never move after time $t$. As soon as $s$ sees $r$ set to $C$, it starts moving toward it (after turning $A$, if $s$ is also set to $C$ and not coincident with $r$), covering at least $\delta$ at each \move phase, until their distance becomes less than $\delta$ and $s$ finally reaches $r$. Then $s$ will turn $C$ as well, and both robots will terminate correctly after seeing each other again.
\end{proof}

\begin{lemma}\label{r5}
If, at some time $t$, the two robots are set to $A$ and $B$ respectively, and neither of them is in a \compute phase that will change its color, they will eventually gather and terminate correctly.
\end{lemma}
\begin{proof}
If some robot ever turns $C$ after time $t$, gathering and termination are ensured by \rrf. Otherwise, the two robots keep seeing each other set to opposite colors, and hence they never change color. The $B$-robot will eventually stay still, and the $A$-robot will then approach it by at least $\delta$ at every \move phase, until their distance is less than $\delta$, and they gather. The $B$-robot then turns $C$, and \rrf applies again.
\end{proof}

Let $r(t)$ denote the position of robot $r$ at time $t\geqslant 0$.

\begin{lemma}\label{r8}
Let $t$ be a time instant at which both robots are set to $A$, and neither of them is in a \compute phase. Let us assume that robot $r$ will stay still until the end of its current phase (even if it is a \move phase), and that robot $s$ will either stay still until the end of its current phase, or its destination point is $r$'s current location. Then $r$ and $s$ will eventually gather and terminate correctly.
\end{lemma}
\begin{proof}
If $s$ is not directed toward $r$ at time $t$, let $d$ be the distance between $r(t)$ and $s(t)$. Otherwise, let $t'$ be the time at which $s$ performed its last \look, and let $d$ be the distance between $s(t')$ and $r(t)$. Furthermore, let $k=\lceil d/\delta\rceil$. We will prove our claim by well-founded induction on $k$, so let us assume our claim to hold for every $k'$ such that $0\leqslant k'<k$.

The first robot to perform a \look after time $t$ sees the other robot set to $A$. If they coincide (i.e., if $s$ has reached $r$ or if $k=0$), the first robot turns $C$, and \rrf applies. If they do not coincide, the first robot turns $B$. If it turns $B$ before the other robot has performed a \look, then \rre applies. Otherwise, when the second robot performs its first \look after time $t$, it sees the first robot still set to $A$. Once again, if they coincide, the second robot turns $C$ and \rrf applies. At this point, if $k=1$ and $s$ was directed toward $r$ at time $t$, the robots have gathered and terminated correctly.

Hence, if $r$ and $s$ perform their first \look at times $t_r$ and $t_s$ respectively, we may assume that both will turn $B$, $r$ computes the midpoint $m_r$ of $r(t_r)$ and $s(t_r)$, and $s$ computes the midpoint $m_s$ of $r(t_s)$ and $s(t_s)$. Observe that, if $s$'s destination was not $r$ at time $t$, then $m_r=m_s$.

Without loss of generality, let $r$ be the first robot to perform the second \look. $r$ sees $s$ set to $B$, hence it turns $A$. If $s$ performs the second \look after $r$ has already turned $A$, then $s$ necessarily sees $r$ in $A$ (because $r$ keeps seeing $s$ in $B$), and \rre applies.

Otherwise, both robots see each other in $B$, and both eventually turn $A$. Without loss of generality, let $s$ be the first robot to perform the third \look. If $k=1$ and $s$ was not directed toward $r$ at time $t$, the robots have indeed gathered in $m_r=m_s$, so $s$ turns $C$ and \rrf applies.

At this point we may assume that $k\geqslant 2$, hence $\delta\leqslant (k-1)\delta<d\leqslant k\delta$. We claim that the distance $d'$ between $r$ and $s$ is now at most $(k-1)\delta$. Indeed, if $s$ was not directed toward $r$ at time $t$, then each robot has either reached $m_r=m_s$, or has approached it by at least $\delta$. In any case, $d'\leqslant d-\delta \leqslant (k-1)\delta$. Otherwise, if $s$ was directed toward $r$ at time $t$, then observe that both $m_r$ and $m_s$ lie between $r(t)$ and $m=(r(t)+s(t))/2$. Moreover, $s$ has performed its first \look while at distance at most $d-\delta$ from $r(t)$, and subsequently it has further approached $r(t)$. On the other hand, $r$ is found between $r(t)$ and $m$, thus at distance not greater than $d/2$ from $r(t)$. Hence, $d'\leqslant \max\{d-\delta, d/2\}\leqslant (k-1)\delta$.

Now, $s$ is the first robot to perform the third \look, and sees $r$ either already in $A$ or still in $B$. In the first case, the inductive hypothesis applies, because $r$ is not in a \compute phase, and its destination is $r$ itself. In the second case, $s$ computes $r$'s location, and it keeps doing so until $r$ turns $A$. When this happens, the inductive hypothesis applies again.
\end{proof}

\begin{corollary}\label{r10}
If, at some time $t$, both robots are set to $B$ and are both in a \wait or in a \look phase, they will eventually gather and terminate correctly.
\end{corollary}
\begin{proof}
The reasoning in the proof of \rrh also implicitly addresses this case. Indeed, the configuration in which both robots are set to $B$ and in a \wait or a \look phase is reached during the analysis, and is incidentally resolved, as well.
\end{proof}

\begin{theorem}\label{r9}
Algorithm~\ref{alg2} solves \rend in non-rigid \asynch and terminates correctly, regardless of the colors in the initial configuration.
\end{theorem}
\begin{proof}
If both robots start in $A$, \rrh applies. If they both start in $B$, \rrj applies. If one robot starts in $A$ and the other one starts in $B$, then \rre applies.

If exactly one robot starts in $C$, it will stay still forever, and the other robot will eventually reach it, turn $C$ as well, and both will terminate.

If both robots start in $C$ and they are coincident, they will terminate. If they are not coincident, let $r$ be the first robot to perform a \look. $r$ will then turn $A$ and move toward the other robot $s$. If $s$ performs its first \look when $r$ has already turned $A$, it will wait, $r$ will eventually reach it, turn $C$, and both will terminate. Otherwise, $s$ performs its \look when $r$ is still set to $C$, hence $s$ will turn $A$ as well, and move toward $r$.

Then, one robot will keep staying $A$ and moving toward the other one, until both have turned $A$. Without loss of generality, let $r$ be the first robot to see the other one set to $A$. If they are coincident, $r$ turns $C$ and \rrf applies. Otherwise, $r$ turns $B$. If this happens before $s$ has seen $r$ in $A$, then \rre applies. Otherwise, both robots will turn $B$. As long as only one robot has turned $B$, it stays $B$ and does not move. At some point, one robot sees the other in $B$ and \rrj applies.
\end{proof}

\section{Impossibility of rendezvous with two colors}\label{s4}

Observe that Algorithms~\ref{alg1},~\ref{alg3} and~\ref{alg2} only produce moves of three types: stay still, move to the midpoint, and move to the other robot. It turns out that, regardless of the number of available colors, any algorithm for \rend must use those three moves under some circumstances.

\begin{proposition}\label{r12}
For any algorithm solving \rend in rigid \fsynch, there exist a color $X$ and a distance $d>0$ such that any robot set to $X$ that sees the other robot at distance $d$ and set to $X$ moves to the midpoint.
\end{proposition}
\begin{proof}
Assume both robots start with the same color and in distinct positions. We may assume that both robots get isometric snapshots at each cycle, so they both turn the same colors, and compute destination points that are symmetric with respect to their midpoint. If they never compute the midpoint and their execution is rigid and fully synchronous, they never gather.
\end{proof}

\begin{proposition}\label{r21}
For any algorithm solving \rend in rigid \ssynch, there exist two colors $X$ and $Y$ and a distance $d>0$ such that any robot set to $X$ that sees the other robot at distance $d$ and set to $Y$ moves to the other robot's position.
\end{proposition}
\begin{proof}
We activate one robot on even cycles, and the other robot on odd cycles. If no robot ever computes the other robot's position and they perform rigid movements, they never gather.
\end{proof}

\begin{proposition}\label{r22}
For any algorithm solving \rend in rigid \ssynch, there exist two colors $X$ and $Y$ and a distance $d>0$ such that any robot set to $X$ that sees the other robot at distance $d$ and set to $Y$ does not move.
\end{proposition}
\begin{proof}
We keep activating only one robot at each cycle (alternately), except when one robot computes the other robot's position. Whenever this happens, we activate both robots for that cycle. If no robot ever performs a null move, they never gather.
\end{proof}

The above observations partly justify the choice to restrict our attention to a specific class of algorithms: from now on, every algorithm we consider computes only destinations of the form
$$(1-\lambda)\cdot me.position + \lambda \cdot other.position,$$
where the parameter $\lambda\in \mathbb R$ depends only on $me.light$ and $other.light$. Similarly, a robot's next light color depends only on the current colors of the two robots' lights, and not on their distance. Recall from Section~\ref{s1} that this class of algorithms is denoted by $\mathcal L$. Notice that Algorithms~\ref{alg1} and~\ref{alg3} both belong to $\mathcal L$, but Algorithm~\ref{alg2} does not, because it may output a different color depending if the two robots coincide or not.

A statement of the form $X(Y)=(Z,\lambda)$ is shorthand for \qq{if a robot is set to $X$ and sees the other robot set to $Y$, it turns $Z$ and makes a move with parameter $\lambda$}, where $\{X,Y,Z\}\subseteq\{A,B\}$ and $\lambda\in \mathbb R$. The negation of $X(Y)=(Z,\lambda)$ will be written as $X(Y)\neq(Z,\lambda)$, wheres a transition with an unspecified move parameter will be denoted by $X(Y)=(Z,\star)$.

\subsection{Preliminary results}

Here we assume that the model is rigid \asynch, that only two colors are available, namely $A$ and $B$, and that the initial configuration is with both robots set to $A$. All our impossibility results for this very special model are then applicable to both non-rigid \asynch with preset initial configuration and rigid \asynch with arbitrary initial configuration.

So, let an algorithm that solves \rend in this model be given. If the algorithm belongs to class $\mathcal L$, then the following statements hold.

\begin{lemma}\label{r13}
$A(A)=(B,\star)$.
\end{lemma}
\begin{proof}
If the execution starts with both robots in $A$, and $A(A)=(A,\star)$, then no robot ever transitions to $B$, and \rend is not solvable, due to \rrn.
\end{proof}

\begin{lemma}\label{r16}
If $A(A)=(B,\oneh)$, then $B(A)=(B,\star)$.
\end{lemma}
\begin{proof}
Let us assume by contradiction that $B(A)=(A,\star)$. If $B(A)=(A,\lambda)$ with $\lambda\neq 1$, we let the two robots execute two cycles each, alternately. As a result, each robot keeps seeing the other robot in $A$, and their distance is multiplied by $|1-\lambda|/2\neq 0$ at every turn. Hence the robots never gather.

If $B(A)=(A,1)$, we let robot $r$ perform a whole cycle and the \look and \compute phases of the next cycle, while the other robot $s$ waits. At this point, their distance has halved, $r$ is set to $A$, and is about to move to $s$'s position. Now $s$ performs two whole cycles, reaching $r$'s position with its light set to $A$. Finally, we let $r$ finish its cycle. As a result, the distance between the two robots has halved, both robots have performed at least a cycle, they are in a \wait phase, and they are both set to $A$. Hence, by repeating the same pattern of moves, they never gather.
\end{proof}

\begin{lemma}\label{r19}
If $A(A)=(B,\oneh)$ and $B(B)=(A,\star)$, then $B(B)=(A,0)$.
\end{lemma}
\begin{proof}
Assume by contradiction that $A(A)=(B,\oneh)$ and $B(B)=(A,\lambda)$ with $\lambda\neq 0$. We let both robots perform a \look and a \compute phase simultaneously. Both turn $B$ and compute the midpoint $m$. Then we let robot $r$ finish the current cycle and perform a new \look. As a result, $r$ will turn $A$ and will move away from $m$. Now let the other robot $s$ finish its first cycle and perform a whole new cycle. $s$ reaches $m$, sees $r$ still set to $B$ and still in $m$, hence $s$ turns $A$ and stays in $m$. Finally, we let $r$ finish the current cycle. At this point, both robots are set to $A$, they are in a \wait phase, both have performed at least one cycle, and their distance has been multiplied by $|\lambda|/2\neq 0$. Therefore, by repeating the same pattern of moves, they never gather.
\end{proof}

\begin{lemma}\label{r17}
If $A(A)=(B,\oneh)$ and $B(B)=(A,0)$, then $B(A)=(B,0)$.
\end{lemma}
\begin{proof}
By \rrp, $B(A)=(B,\star)$. Assume by contradiction that $B(A)=(B,\lambda)$ with $\lambda\neq 0$. We let both robots perform a \look simultaneously, so both plan to turn $B$ and move to the midpoint $m$. We let robot $r$ finish the cycle, while the other robot $s$ waits. Then we let $r$ perform a whole other cycle. So $r$ sees $s$ still in $A$, and moves away from $m$, while staying $B$. Now we let $s$ finish its first cycle and move to $m$. Finally, we let both robots perform a new cycle simultaneously. As a result, both robots are set to $A$ and are in a \wait phase, both have performed at least one cycle, and their distance has been multiplied by $|\lambda|/2\neq 0$. By repeating the same pattern of moves, they never gather.
\end{proof}

\begin{lemma}\label{r18}
If $A(A)=(B,\oneh)$ and $B(B)=(A,0)$, then $A(B)=(A,1)$.
\end{lemma}
\begin{proof}
Let us first assume that $A(B)=(B,\lambda)$ with $\lambda\neq 1$. We let one robot perform a whole cycle, thus turning $B$ and moving to the midpoint. Then we let the other robot perform a cycle, at the end of which both robots are set to $B$. Finally, we let both robots perform a cycle simultaneously, after which they are back to $A$ and in a \wait phase. Because their distance has been multiplied by $|1-\lambda|/2\neq 0$, by repeating the same pattern of moves they never gather.

Assume now that $A(B)=(B,1)$. We let robot $r$ perform a \look and a \compute phase, thus turning $B$ and computing the midpoint. Now we let the other robot $s$ perform a whole cycle, at the end of which it is set to $B$ and has reached $r$. Then we let $r$ finish its cycle, moving away from $s$. Finally, we let both robots perform a new cycle simultaneously, which takes them back to $A$. Their distance has now halved, and by repeating the same pattern of moves they never gather.

Assume that $A(B)=(A,\star)$, and let robot $r$ perform an entire cycle, thus turning $B$ and moving to the midpoint. Due to \rrq, $B(A)=(B,0)$, which means that, from now on, both robots will retain colors. Hence, $r$ will always stay still, and $s$ will never reach $r$ unless $A(B)=(A,1)$.
\end{proof}

\subsection{Rigid asynchronous model with arbitrary initial configuration}

\begin{lemma}\label{r23}
Algorithm~\ref{alg1} does not solve \rend in rigid \asynch, if both robots are set to $B$ in the initial configuration.
\end{lemma}
\begin{proof}
Let both robots perform a \look phase, so that both will turn $A$. We let robot $r$ finish the current cycle and perform a new \look, while the other robot $s$ waits. Hence, $r$ will stay $A$ and move to $s$'s position. Now we let $s$ finish the current cycle and perform a new \look. So $s$ will turn $B$ and move to the midpoint $m$. We let $r$ finish the current cycle, thus reaching $s$, and perform a whole new cycle, thus turning $B$. Finally, we let $s$ finish the current cycle, thus turning $B$ and moving to $m$. As a result, both robots are again set to $B$, they are in a \wait phase, both have executed at least one cycle, and their distance has halved. Thus, by repeating the same pattern of moves, they never gather.
\end{proof}

\begin{theorem}\label{r15}
There is no algorithm of class $\mathcal L$ that solves \rend using two colors in rigid \asynch from all possible initial configurations.
\end{theorem}
\begin{proof}
Because robots may start both in $A$ or both in $B$, the statement of \rrm, holds also with $A$ and $B$ exchanged. Hence $A(A)=(B,\star)$, but also $B(B)=(A,\star)$. Moreover, by \rrl, either $A(A)=(B,\oneh)$ or $B(B)=(A,\oneh)$. By symmetry, we may assume without loss of generality that $A(A)=(B,\oneh)$. Now, by \rrs, $B(B)=(A,0)$. Additionally, by \rrq and \rrr, $B(A)=(B,0)$ and $A(B)=(A,1)$. These rules define exactly Algorithm~\ref{alg1}, which is not a solution, due to \rrw.
\end{proof}

\subsection{Non-rigid asynchronous model with preset initial configuration}

\begin{theorem}\label{r11}
There is no algorithm of class $\mathcal L$ that solves \rend using two colors in non-rigid \asynch, even assuming that both robots are set to a predetermined color in the initial configuration.
\end{theorem}
\begin{proof}
Let both robots be set to $A$ in the initial configuration, and let $d\geqslant 0$ be given. By \rrm, $A(A)=(B,\lambda)$, for some $\lambda\in\mathbb R$. If $\lambda\neq \oneh$, we place the two robots at distance $d/|1-2\lambda|$ from each other, and we let them perform a whole cycle simultaneously. If $\lambda=\oneh$, we place the robots at distance $d+2\delta$, and we let them perform a cycle simultaneously, but we stop them as soon as they have moved by $\delta$. As a result, both robots are now set to $B$, and at distance $d$ from each other. This means that any algorithm solving \rend with both robots set to $A$ must also solve it with both robots set to $B$, as well.

Similarly, we can place the two robots at distance $d/|1-\lambda|$ or $d+\delta$, depending if $\lambda\neq 1$ or $\lambda=1$. Then we let only one robot perform a full cycle, and we let it finish or we stop it after $\delta$, in such a way that it ends up at distance exactly $d$ from the other robot. At this point, one robot is set to $A$ and the other is set to $B$.

It follows that any algorithm for \rend must effectively solve it from all possible initial configurations. But this is impossible, due to \rro.
\end{proof}

\section{Conclusions}\label{s5}

We considered deterministic distributed algorithms for \rend for mobile robots that cannot use distance information, but can only reduce (or increase) their distance by a constant factor, depending on the color of the lights that both robots are carrying. We called this class of algorithms $\mathcal L$.

We gave several upper and lower bounds on the number of different colors that are necessary to solve \rend in different robot models. Based on these results, we can now give a complete characterization of the number of necessary colors in every possible model, ranging from fully synchronous to semi-synchronous to asynchronous, rigid and non-rigid, with preset or arbitrary initial configuration.

\begin{theorem}\label{r20}
To solve \rend with an algorithm of class $\mathcal L$ from a preset starting configuration,
\begin{itemize}
\item one color is sufficient for rigid and non-rigid \fsynch;
\item two colors are necessary and sufficient for rigid \ssynch, non-rigid \ssynch, and rigid \asynch;
\item three colors are necessary and sufficient for non-rigid \asynch.
\end{itemize}
To solve \rend with an algorithm of class $\mathcal L$ from an arbitrary starting configuration,
\begin{itemize}
\item one color is sufficient for rigid and non-rigid \fsynch;
\item two colors are necessary and sufficient for rigid and non-rigid \ssynch;
\item three colors are necessary and sufficient for rigid and non-rigid \asynch.
\end{itemize}
\end{theorem}
\begin{proof}
All the optimal color values derive either from previous theorems or from the model inclusions summarized in Figure~\ref{f0}.

That just one color is (necessary and) sufficient for all \fsynch models follows from \rrn.

\rrn also implies that, for all the other models, at least two colors are necessary. Therefore, by \rra, two colors are necessary and sufficient for all \ssynch models.

Similarly, \rrb states that two colors are necessary and sufficient for rigid \asynch with preset initial configuration. On the other hand, by \rro and \rrk, three colors are necessary in the three remaining models, and by \rrx three colors are also sufficient.
\end{proof}

In the three models in which three colors are necessary and sufficient, it remains an open problem to determine whether using distance information to its full extent would make it possible to use only two colors.

An interesting variation on this model is when the light on a robot can be seen only by the other robot(s). In this case, algorithms of class $\mathcal L$ are inadequate to solve \rend even in rigid \asynch with preset initial configuration, regardless of the number of available colors. In contrast, three colors are necessary and sufficient for all \ssynch models.

On the other hand, if the light is visible only to the robot that is carrying it (i.e., internal memory), then no algorithm of class $\mathcal L$ can solve \rend, even in rigid \ssynch with preset initial configuration, regardless of the number of colors.

\end{document}